\documentclass[11pt]{article}

\usepackage{amsmath}
\usepackage{amsthm}
\usepackage{graphicx} 
\usepackage{array} 
\usepackage{bbm}
\usepackage{thm-restate}

\usepackage{fancybox}

\newenvironment{fminipage}%
  {\begin{Sbox}\begin{minipage}}%
  {\end{minipage}\end{Sbox}\fbox{\TheSbox}}

\newenvironment{algbox}[0]{\vskip 0.2in
\noindent 
\begin{fminipage}{6.3in}
}{
\end{fminipage}
\vskip 0.2in
}

\newcommand{\norm}[1]{\left\lVert#1\right\rVert}

\usepackage{amsmath, amssymb, amsfonts, verbatim}
\usepackage{hyphenat,epsfig,subfigure,multirow}
\usepackage{multicol}
\usepackage{nicefrac}
\usepackage{paralist}
\usepackage{enumitem}
\usepackage{tikz}
\usetikzlibrary{positioning}

\usepackage[ruled]{algorithm2e}

\DeclareFontFamily{U}{mathx}{\hyphenchar\font45}
\DeclareFontShape{U}{mathx}{m}{n}{
      <5> <6> <7> <8> <9> <10>
      <10.95> <12> <14.4> <17.28> <20.74> <24.88>
      mathx10
      }{}
\DeclareSymbolFont{mathx}{U}{mathx}{m}{n}
\DeclareMathSymbol{\bigtimes}{1}{mathx}{"91}

\usepackage{tcolorbox}
\tcbuselibrary{skins,breakable}
\tcbset{enhanced jigsaw}

\usepackage[normalem]{ulem}
\usepackage[compact]{titlesec}

\definecolor{DarkRed}{rgb}{0.5,0.1,0.1}
\definecolor{DarkBlue}{rgb}{0.1,0.1,0.5}

\usepackage{nameref}
\definecolor{ForestGreen}{rgb}{0.1333,0.5451,0.1333}
\definecolor{Red}{rgb}{0.9,0,0}
\usepackage[linktocpage=true,
	pagebackref=true,colorlinks,
	linkcolor=DarkRed,citecolor=ForestGreen,
	bookmarks,bookmarksopen,bookmarksnumbered]
	{hyperref}
\usepackage[noabbrev,nameinlink]{cleveref}
\crefname{property}{property}{Property}
\creflabelformat{property}{(#1)#2#3}
\crefname{equation}{eq}{Eq}
\creflabelformat{equation}{(#1)#2#3}

\usepackage{bm}
\usepackage{url}
\usepackage{xspace}
\usepackage[mathscr]{euscript}

\usepackage{tikz}
\usetikzlibrary{arrows}
\usetikzlibrary{arrows.meta}
\usetikzlibrary{shapes}
\usetikzlibrary{backgrounds}
\usetikzlibrary{positioning}
\usetikzlibrary{decorations.markings}
\usetikzlibrary{patterns}
\usetikzlibrary{calc}
\usetikzlibrary{fit}
\usetikzlibrary{snakes}

\usepackage{mdframed}

\usepackage[noend]{algpseudocode}
\makeatletter
\def\BState{\State\hskip-\ALG@thistlm}
\makeatother

\usepackage{cite}
\usepackage{enumitem}

\usepackage[margin=1in]{geometry}

\newtheorem{theorem}{Theorem}[section]
\newtheorem{example}{Example}[section]
\newtheorem{lemma}{Lemma}[section]

\newtheorem{corollary}[lemma]{Corollary}

\newtheorem{definition}[lemma]{Definition}

\DeclareMathOperator{\sign}{sgn}

\newtheorem*{theorem*}{Theorem}
\newtheorem*{claim*}{Claim}
\newtheorem*{proposition*}{Proposition}
\newtheorem*{lemma*}{Lemma}
\newtheorem*{problem*}{Problem}

\crefname{lemma}{Lemma}{Lemmas}
\crefname{claim}{Claim}{Claims}

\newtheorem{mdresult}{Result}

\newtheoremstyle{restate}{}{}{\itshape}{}{\bfseries}{~(restated).}{.5em}{\thmnote{#3}}
\theoremstyle{restate}

\allowdisplaybreaks

\renewcommand{\qed}{\nobreak \ifvmode \relax \else
      \ifdim\lastskip<1.5em \hskip-\lastskip
      \hskip1.5em plus0em minus0.5em \fi \nobreak
      \vrule height0.75em width0.5em depth0.25em\fi}

\renewcommand{\leq}{\leqslant}
\renewcommand{\geq}{\geqslant}

\title{
Entrywise Approximate Laplacian Solving
}
\author{
Jingbang Chen\footnote{jingbang.chen@uwaterloo.ca} \and Mehrdad Ghadiri\footnote{mehrdadg@mit.edu} \and Hoai-An Nguyen\footnote{hnnguyen@andrew.cmu.edu} \and Richard Peng\footnote{yangp@cs.cmu.edu} \and Junzhao Yang\footnote{junzhaoy@cs.cmu.edu}
}
\date{}

\newcommand\ee{\boldsymbol{\mathit{e}}}
\newcommand\xx{\boldsymbol{\mathit{x}}}
\newcommand\bb{\boldsymbol{\mathit{b}}}
\newcommand\DD{\boldsymbol{\mathit{D}}}
\renewcommand\AA{\boldsymbol{\mathit{A}}}
\newcommand\BB{\boldsymbol{\mathit{B}}}
\newcommand\CC{\boldsymbol{\mathit{C}}}

\newcommand\NN{\boldsymbol{\mathit{N}}}
\newcommand\TT{\boldsymbol{\mathit{T}}}
\renewcommand\SS{\boldsymbol{\mathit{S}}}

\newcommand\II{\boldsymbol{\mathit{I}}}
\newcommand\LL{\boldsymbol{\mathit{L}}}
\newcommand\MM{\boldsymbol{\mathit{M}}}
\newcommand\XX{\boldsymbol{\mathit{X}}}
\newcommand\YY{\boldsymbol{\mathit{Y}}}
\newcommand\vv{\boldsymbol{\mathit{v}}}
\newcommand\zz{\boldsymbol{\mathit{z}}}
\newcommand\uu{\boldsymbol{\mathit{u}}}
\newcommand\pp{\boldsymbol{\mathit{p}}}

\newcommand\MMtil{\widetilde{\boldsymbol{\mathit{M}}}}
\newcommand\ZZ{\boldsymbol{\mathit{Z}}}
\newcommand\pr{\textbf{Pr}}

\renewcommand{\sc}{\textsc{Sc}}
\newcommand\N{\mathbb{N}}
\renewcommand\P{\mathbb{P}}
\newcommand\E{\mathbb{E}}
\newcommand{\vertiii}[1]{{\left\vert\kern-0.25ex\left\vert\kern-0.25ex\left\vert #1 
    \right\vert\kern-0.25ex\right\vert\kern-0.25ex\right\vert}}

\newcommand{\approxbar}{\overline{\approx}}
\newcommand{\Otil}{\widetilde{O}}

\newcommand{\sink}{\textsc{Sink}}
\newcommand{\R}{\mathbb{R}}

\begin{document}

\maketitle

\begin{abstract}
We study the escape probability problem in random walks over graphs. 
Given vertices, $s,t,$ and $p$, the problem asks for the probability that a random walk starting at $s$ will hit $t$ before hitting $p$. 
Such probabilities can be exponentially small even for unweighted undirected graphs with polynomial mixing time. 
Therefore current approaches, which are mostly based on fixed-point arithmetic, require $n$ bits of precision in the worst case. 

We present algorithms and analyses for weighted directed graphs under floating-point arithmetic and improve the previous best running times in terms of the number of bit operations.
We believe our techniques and analysis could have a broader impact on the computation of random walks on graphs both in theory and in practice. 
\end{abstract}

\clearpage

\section{Introduction}
Solving linear systems is the workhorse of the modern approach to optimization. 
There has been a significant effort in the past two decades to design more efficient algorithms for structured linear systems \cite{ST04:journal,PV21,FFG22}. 
The best example of these efforts is the near-linear time solvers for Laplacian systems that unlocked many fast algorithms for graph problems ranging from the computation of different probabilities associated with random walks (Markov chains) \cite{CKPPRSV17} to network flow problems \cite{CKLPGS22}.

Despite the significant progress towards such algorithms, the practical usage of these algorithms is poorly understood.
Perhaps the main reason is that many of these algorithms are analyzed under unrealistic assumptions and number systems such as exact arithmetic (real-RAM) and fixed-point arithmetic. 
The former is a model of computation that assumes arithmetic operations can be performed to infinite precision in constant time. 
The latter is a more realistic assumption that considers finite precision, but it is still widely different from how real computers operate. 
In this paper, we initiate the study of solving Laplacian linear systems under floating-point arithmetic.

A major motivation for our study is the computation of probabilities associated with random walks on graphs. 
As we will see in \Cref{example:bad-example}, such probabilities can be exponentially small even for seemingly nice unweighted undirected graphs. 
Note that to even store a number $c$ which is around $2^{-n}$ using fixed-point numbers, we need about $O(n)$ time and bits of memory.
However, with floating point numbers, we can store a number that is within a $e^{\epsilon}$ (multiplicative) factor of $c$ with only $O(\log n + \log(1/\epsilon))$ time and bits of memory. 
The $\log n$ factor is to store the exponent and $\log(1/\epsilon)$ is to store the required bits of precision.
Therefore when working with large and small numbers, it is more efficient to use floating-point numbers.
However, the stability analysis of floating-point arithmetic is often very complicated.
Even seemingly trivial operations such as adding $n$ numbers together can have prohibitive errors that prevent an algorithm from running successfully --- see Lectures 14 and 15 in \cite{TB97:book}.
This is the motivation behind methods such as Kahan's summation algorithm, which reduces the error for adding numbers \cite{K65}.

In this paper, we study the escape probability problem: the probability of a random walk starting at vertex $s$ in a graph to hit vertex $t$ before hitting vertex $p$.

\begin{restatable}{theorem}{mainTheorem}
\label{thm:main-theorem}
    Given a weighted directed graph $G=(V,E)$ with $n$ vertices and nonnegative edge weights that are given with $L$ bits in floating-point, and $t,p\in V$, there is an algorithm that computes the escape probability for all starting vertices $s\in V$ within an $e^{\epsilon}$ multiplicative factor with $\Otil\left(n^3 \cdot \left(L + \log \frac{1}{\epsilon}\right)\right)$ bit operations.
\end{restatable}

In \Cref{thm:main-theorem}, the $L$ bits refer to both bits of precision and the bits required for the exponent of the floating point number. 
One might think that the near-linear time algorithms for solving Laplacian systems (and generally diagonally-dominant systems) or approaches based on fast-matrix-multiplication achieve a better running time than \Cref{thm:main-theorem}. 
However, as we will discuss extensively in \Cref{sec:discussion-and-related}, these approaches often only count the number of arithmetic operations (not bit operations), and due to the nature of their error (which is bounded norm-wise), they require a significantly larger number of bit operations compared to \Cref{thm:main-theorem}.
Namely, the near-linear-time approaches and fast-matrix-multiplication approaches require about $mn^2$ and $n^{\omega+1}$ bit operations, respectively, where $m$ is the number of edges.

Our main tool in proving \Cref{thm:main-theorem} is an algorithm that given a row diagonally dominant $L$-matrix (RDDL) matrix (see \Cref{def:rddl}) $\NN$ produces a matrix $\ZZ$ such that for all $i,j$, $e^{-\epsilon} \ZZ_{ij} \leq \NN^{-1}_{ij} \leq e^{\epsilon} \ZZ_{ij}$, which we denote with $\NN^{-1} \approxbar_{\epsilon} \ZZ$. Note that not all RDDL matrices are invertible. For example, Laplacian matrices are RDDL and are not invertible --- they have the vector of all ones in their kernel.

\begin{restatable}{theorem}{recursionTheorem}
    \label{thm:main-recursion}
    Let $\MM \in \R^{n\times n}$ be an RDDL matrix and $\vv\in \R_{\leq 0}^{n}$ with at least one entry of $\vv$ being strictly less than zero. Suppose the entries of $\MM$ and $\vv$ are presented with $L$ bits in floating-point. Let $\NN\in\R^{n\times n}$ such that for $i\neq j$, $\NN_{ij}=\MM_{ij}$, and for $i\in[n]$, $\NN_{ii} = -(\vv_i + \sum_{j\in[n]\setminus\{i\}} \MM_{ij})$. If $\NN$ is invertible, then $\textsc{RecInvert}(\MM, \vv, \epsilon)$ (see \Cref{fig:RecInvertWithExcess}) returns a matrix $\ZZ$ such that $\ZZ \approxbar_{\epsilon} \NN^{-1}$ with $\Otil\left(n^3 \cdot \left(L + \log \frac{1}{\epsilon}\right)\right)$ bit operations.
\end{restatable}

Note that our notion of \emph{approximate} inverse in \cref{thm:main-recursion} is stronger than usual numerical linear algebraic guarantees that bound the error norm-wise (not entry-wise). 
As we will discuss in the next section, inversion using fast-matrix-multiplication is not capable of providing such entry-wise guarantees. 
Finally, we note that the purpose of the vector $\vv$ in \Cref{thm:main-recursion} is to guarantee that matrix $\NN$ is row diagonally dominant since even checking row diagonal dominance using floating-point numbers is not possible --- we cannot check whether the sum of $n$ numbers is equal to zero or just very close to zero.

The rest of the paper is organized as follows. In \Cref{sec:discussion-and-related}, we discuss other approaches that can be used for computing escape probabilities and why they result in a higher number of bit operations when we want multiplicative error factors. We discuss the notation and preliminaries required for the rest of the paper in \Cref{sec:prelim}. We discuss the linear system we need to solve for computing escape probabilities in \Cref{sec:linear-system}. \Cref{sec:repeated-square} discusses a simple algorithm based on repeated squaring that works for graphs with edge weights in a polynomial range. We believe this algorithm might be especially of practical interest. Finally, in \Cref{sec:recursive-SC}, we prove our main result using a recursive algorithm that uses Schur complement techniques.

\subsection{Discussion and Related Work}
\label{sec:discussion-and-related}

In this section, we discuss different approaches in the literature that can be used to compute the escape probability in a graph. As discussed in \Cref{sec:linear-system}, one can compute the escape probability by solving a linear system. Therefore, different approaches, including Laplacian solvers, fast-matrix-multiplication, and iterative methods, can be used to compute the escape probability. However, we argue that all of these approaches result in a larger running time (in terms of the number of bit operations) for computing (exponentially) small escape probabilities. Such small probabilities are illustrated in the following example.

\begin{example}
\label{example:bad-example}
In the graph illustrated in \Cref{fig:bad-example}, the probability that a random walk starting at $s$ will hit $t$ before $p$ is exponentially small in the number of vertices. The reason is that any walk that does not hit $p$ should only take edges on the path between $s$ and $t$. To see this note that any walk reaching from $s$ to $t$ should traverse at least $n-2$ edges. Therefore the probability that we are still on the $s$-$t$ path after $n-2$ steps is an upper bound for the escape probability. Since with probability at least $1/3$ in each step, we exit the path, the probability that we are still on the path after $n-2$ steps is at most
\[
\sum_{i=n-2}^{\infty} \left(\frac{2}{3}\right)^{i} = 2 \cdot \left( \frac{2}{3} \right)^{n-3}.
\]

\begin{figure}[h]
    \centering
    \begin{tikzpicture}
    \node (hub) at (0, 1) [circle, draw, fill=red!20] {$p$};
    \node (n1) at (-3, 0) [circle, draw, fill=blue!20] {$s$};
    \node (n2) at (-1.5, 0) [circle, draw, fill=blue!20] {};
    \node (n3) at (0, 0) [circle, draw, fill=blue!20] {};
    \node (n4) at (1.5, 0) [circle, draw, fill=blue!20] {};
    \node (n5) at (3, 0) [circle, draw, fill=blue!20] {$t$};
    \node (n6) at (0.75, 0) [] {$\cdots$};

    \draw (n1) -- (n2);
    \draw (n2) -- (n3);
    \draw (n4) -- (n5);
    \draw (n3) -- (0.4, 0);
    \draw (1.1, 0) -- (n4);

    \draw (hub) -- (n1);
    \draw (hub) -- (n2);
    \draw (hub) -- (n3);
    \draw (hub) -- (n4);
    \draw (hub) -- (n5);
\end{tikzpicture}  
    \caption{Exponentially small escape probability.}
    \label{fig:bad-example}
\end{figure}
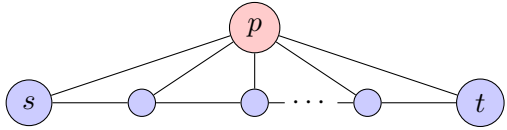
  
\end{example}

\paragraph{Near-linear-time Laplacian solvers.} 
Spielman and Teng \cite{ST04:journal} have shown that for a symmetric diagonally dominant matrix $\LL$, one can approximately solve the linear system $\LL \xx = \bb$ with $O(m \log^c n \log \epsilon^{-1})$ arithmetic operations on numbers with $O(\log(\kappa(\LL)) \log^c n \log \epsilon^{-1})$ bits of precision. Here, $\kappa(\LL)$ is the condition number of $\LL$ and $c$ is a constant. More specifically, their algorithm outputs $\widetilde{\xx}$ such that $\norm{\widetilde{\xx} - \LL^{\dagger} \bb}_{\LL} \leq \epsilon \cdot \norm{\LL^{\dagger} \bb}_{\LL}$, where $\LL^{\dagger}$ is the pseudo-inverse of $\LL$.
By taking $\epsilon'$ to be smaller than $\epsilon/(\kappa(\LL))^2$ and orthogonalizing against the all-$1$s vector on each connected component,
we get $\norm{\widetilde{\xx} - \LL^{\dagger} \bb}_{2} \leq \epsilon' \cdot \norm{\LL^{\dagger} \bb}_{2}$.
This bounds the norm-wise error of the computed vector of escape probabilities. The next example shows why norm-wise error bounds are not suited for computing small escape probabilities.

\begin{example}
\label{ex:norm-wise-bad}
Let $\vv=(10^{15},1)$, $\widetilde{\vv}=(10^{15}(1+\epsilon),0)$, and $\epsilon=10^{-5}$. Note that the second entry of $\vv$ and $\widetilde{\vv}$ are very different from each other.
Then
\begin{align}
\frac{\norm{\widetilde{\vv} - \vv}_1}{\norm{\vv}_1} = \frac{\epsilon\cdot 10^{15}+1}{10^{15}+1} = \frac{10^{10}+1}{10^{15}+1} \approx 10^{-5}.
\end{align}
This means that $\widetilde{\vv}$ approximates $\vv$ in a norm-wise manner, but entry-wise $\vv$ and $\widetilde{\vv}$ are very different vectors.
Therefore a bound on the norm does not necessarily provide bounds for the entries. For the norm bound to give guarantees for the entries, we would need to have $\epsilon \approx \frac{\vv_{\min}}{\vv_{\max}}$, where $\vv_{\min}$ and $\vv_{\max}$ are the smallest and largest entry of $\vv$ in terms of absolute value, respectively.
\end{example}

Now note that for the graph in \Cref{example:bad-example}, we have to take the error parameter $\epsilon$ to be exponentially small (in $n$) to be able to find a multiplicative approximation to the smallest escape probability. This means that the total number of bit operations for the Spielman-Teng algorithm will be $\Otil(m n^2)$. Later works such as \cite{KGP14} that improve the constant $c$ and algorithms for directed Laplacians
\cite{CKPPSV16,CKPPRSV17} all have the same dependencies on $\log(1/\epsilon)$ and similar norm-wise gurantees.

\paragraph{Fast-matrix-multiplication.}

Strassen \cite{S69} has shown that two $n$-by-$n$ matrices can be multiplied with fewer than $n^3$ arithmetic operations. Currently, the best bound for the number of arithmetic operations for fast-matrix-multiplication is $O(n^{\omega})$, where $\omega < 2.372$ \cite{WXXZ24} which is based on techniques from \cite{CW90}. It is also well-known that matrix inversion can be reduced to polylogarithmic matrix multiplications. Therefore a linear system can be solved in $O(n^{\omega})$ arithmetic operations.

The stability of Strassen's algorithm and other fast-matrix-multiplication algorithms have been a topic of debate for decades \cite{M75,BLS91,H90}. Although it is established that such algorithms are stable \cite{DDHK07,DDH07}, similar to near-linear-time Laplacian solvers, such stability only holds in a norm-wise manner. In other words, fast-matrix-multiplication algorithms produce errors that can only be bounded norm-wise. This is observed by Higham \cite{H90}, and he provides an explicit example: 

\begin{align}
C = \begin{bmatrix}
1 & 0 \\ 0 & 1
\end{bmatrix}
\begin{bmatrix}
1 & \epsilon \\ \epsilon & \epsilon^2
\end{bmatrix}.
\end{align}

As Higham \cite{H90} shows, Strassen's algorithm fails in computing $C_{22}$ accurately, and he states that ``to summarize, Strassen’s method has less favorable stability properties than
conventional multiplication in two respects: it satisfies a weaker error bound
(norm-wise rather than component-wise)\ldots The norm-wise bound is a consequence
of the fact that Strassen’s method adds together elements of $\AA$ matrix-wide (and
similarly for $\BB$).''

As Higham states, this issue arises from the subtractions (i.e., adding positive and negative numbers together) in Strassen's algorithm. ``Another interesting property of Strassen’s method is that it always involves
some genuine subtractions (assuming that all additions are of nonzero terms).
This is easily deduced from the formulas (2.2). As noted in \cite{GL89}, this makes
Strassen’s method unattractive in applications where all the elements of $\AA$
and $\BB$ are nonnegative (for example, in Markov processes \cite{H87}). Here, conventional multiplication yields low relative error component-wise because in (4.2)
$| \AA | | \BB | = |\AA \BB| = |\CC|$, yet comparable accuracy cannot be guaranteed for
Strassen’s method.''

Therefore for any algorithm based on fast-matrix-multiplication, we need to take the error parameter exponentially small in $n$. This gives $n^{\omega+1}$ bit operations which is worse than $n^3$.

\paragraph{Shifted and $p$-adic numbers.} 
Another approach for solving linear systems is to use $p$-adic numbers
\cite{S05,D82}.
The advantage of this approach that relies on Cramer's rule is that the solution is computed exactly in rational number representation. In addition, the running time of the algorithms do not have a dependence on error parameters or the condition number of the matrix. However, such algorithms only work with integer input matrices and vectors. In the escape probability problem, either the input matrix or the input vector has to be rational (not integer). Naive ways of rounding such inputs to integers would result in either increasing the bit complexity of the input numbers or giving norm-wise errors (which again would result in algorithms with a running time of $O(n^{\omega+1})$). We believe these algorithms can be adopted to give running times of $O(n^{\omega} L)$, where $L$ is the bit complexity of input numbers in fixed-point representation when the input numbers are rational instead of integer. However, this does not follow immediately. Moreover, our results in \Cref{thm:main-recursion} can handle input numbers that are exponentially large or small in fixed-point representation (but have small floating-point representation). Such input numbers would result in a running time of $O(n^{\omega+1})$ for approaches based on $p$-adic numbers, which is worse than $n^3$. Adopting these algorithms for floating-point numbers, if possible, would require significant modifications.

\paragraph{Gaussian elimination.}

There are empirical studies in the literature that show Gaussian elimination is more stable than other approaches for computing stationary distribution of Markov chains \cite{H87,GTH85,HP84}. It is observed and stated that the process is more stable because it does not require subtractions in this case. This is very similar to our ideas for proving \Cref{thm:main-recursion} using Schur complements. However, to the best of our knowledge, these approaches have not been theoretically studied before our work.


\paragraph{Bit complexity.} There are many recent works that study the bit complexity of algorithms for linear algebraic primitives, such as diagonalization \cite{S23,BVKS23,BVS22a,BVS22b,DKRS23} and optimization \cite{GPV23,G23,ABGZ24,GLPSWY24}. All of these works provide norm-wise error bounds.

\section{Preliminaries}
\label{sec:prelim}


Given a directed weighted graph $G=(V,E, w)$ with $w\in\R_{\geq 0}^{E}$, a random walk in the graph picks the next step independent of all the previous steps. We denote the neighbors of vertex $v$ with $N(v)$. Then if the random walk is at vertex $v$, in the next step it goes to $u\in N(v)$ with probability $\frac{w(v,u)}{\sum_{y \in N(v)} w(v,y)}$. In other words, we consider the Markov chain associated with the graph. Note that for the special case of unweighted graphs the probability for each neighbor is equal to $1/|N(v)|$.
Then the escape probability is defined as the following.

\begin{definition}[Escape Probability]
    The escape probability $\P(s,t,p)$ denotes the probability of hitting the vertex $t$ before $p$ among all the possible random walks starting at $s$.
\end{definition}
Note that there is some symmetry associated with escape probability: $\P(s,t,p) = 1- \P(s,p,t)$. However due to issues that floating points introduce (that we discuss later in this section), it is not advisable to compute $\P(s,t,p)$ from $\P(s,p,t)$ in this way since it involves adding a negative number to a positive number (i.e., subtraction).

We denote the number of vertices and edges with $n$ and $m$, respectively. We use bold small letter and bold capital letters to denote vectors and matrices, respectively. The vector of all ones is denoted with $\boldsymbol{1}$ and the $i$'th standard basis vector is denoted by $\boldsymbol{e}_i$. We do not explicitly show the size of these vectors, but it will be clear from the context throughout the paper. We use $\Otil$ notation to omit polylogarithmic factors in $n$ and $L$ and polyloglog factors in $1/\epsilon$. In other words, $\Otil(f) = O(f \cdot \log(nL\cdot \log(\frac{1}{\epsilon})))$. We rely extensively on the following notion of approximation for scalars and matrices. 

\begin{definition}
For two nonnegative scalars $a$ and $b$, and $\epsilon\geq 1$,
we denote $a \approx_{\epsilon} b$ if
\begin{align*}
e^{-\epsilon} \cdot a
\leq b
\leq e^{\epsilon} \cdot a
\end{align*}
For two matrices $A$ and $B$ of same size, we denote
$
\AA
\approxbar_{\epsilon}
\BB
$
if $\AA_{ij} \approx_{\epsilon} \BB_{ij}$
for all $1 \leq i, j \leq n$.
\end{definition}
Note that for nonnegative numbers $a,b,c$, and $d$ if $a \approx_{\epsilon_1} b$ and $b \approx_{\epsilon_2} c$, then $a \approx_{\epsilon_1 + \epsilon_2} c$, and if $a\approx_{\epsilon} c$ and $b\approx_{\epsilon} d$, then $a+b\approx_{\epsilon} c+d$. We use this strong notion approximation for the inversion of a special class of matrices called RDDL in \Cref{sec:recursive-SC}.

\begin{definition}[RDDL]
\label{def:rddl}
    $\MM\in \R^{n\times n}$ is an $L$-matrix if for all $i\neq j$, $\MM_{ij} \leq 0$, and for all $i\in [n]$, $\MM_{ii} > 0$.
    $\MM$ is row diagonally dominant (RDD) if for all $i\in[n]$, $|\MM_{ii}| \geq \sum_{j\in [n] \setminus\{i\}} |\MM_{ij}|$. $\MM$ is RDDL if it is both an $L$-matrix and RDD.
\end{definition}

We can use an RDDL matrix to show the probabilities associated with a random walk on a weighted directed graph. Setting the set of vertices to $[n]:=\{1,\ldots,n\}$, for $i,j \in [n]$ with $i\neq j$ , if $i$ is connected to $j$, then $\MM_{ij} = -\frac{w(i,j)}{\sum_{k \in N(i)} w(i,k)}$, and $\MM_{ij} = 0$, otherwise. Also, we set $\MM_{ii} = \sum_{j\in [n] \setminus \{i\}}$. Note that if we remove a row and the corresponding column from an RDDL matrix, the matrix stays RDDL. Throughout the paper, we use $V$ and $[n]$ for the set of vertices of the graph interchangeably.

For an RDDL matrix $\MM \in \R^{n \times n}$, we define the corresponding matrix $G = ([n+1], E)$. $n+1$ is a dummy vertex that we add. For any $i,j \in [n]$ with $i\neq j$, the weight of edge $(i,j)$ is $-\MM_{ij}$ in the graph. Moreover, the weight of edge $(i, n+1)$ is equal to $\sum_{j \in [n]} \MM_{ij}$, and the weight of edge $(n+1,i)$ is equal to zero. If in such a graph, for all vertices $i\in[n]$, there is a path from $i$ to $n+1$ with all positive weight edges, then we say $G$ (the graph corresponding to the RDDL matrix) is \emph{connected to the dummy vertex}.

The $\ell_1$-norm and $\ell_{\infty}$ norm of a vector $\vv \in \R^{n}$ are $\norm{\vv}_1 := \sum_{i\in [n]} |\vv_i|$ and $\norm{\vv}_{\infty} := \max_{i\in [n]} |\vv_i|$ respectively. Then the induced $\ell_1$ and $\ell_{\infty}$ norm is defined as the following for a matrix $\MM \in \R^{n\times n}$.
\[
\norm{\MM}_1 := \max_{\vv \in \R^{n}: \norm{\vv}_1 = 1} \norm{\MM \vv}_1, \text{ and } \norm{\MM}_{\infty} := \max_{\vv \in \R^{n}: \norm{\vv}_{\infty} = 1} \norm{\MM \vv}_{\infty}
\]
One can easily see that both the norms and the induced norms satisfy triangle inequality and consistency, i.e., $\norm{\AA \BB}_{\infty} \leq \norm{\AA}_{\infty} \cdot \norm{\BB}_{\infty}$. The spectral radius of matrix $\MM$ is denoted by $\rho(\MM) := \max\{|\lambda_1|,\ldots,|\lambda_n|\}$, where $\lambda_i$'s are the eigenvalues of $\MM$.

\paragraph{Schur complement.} Given a block matrix
\[
\MM = \begin{bmatrix}
    \AA & \BB \\ \CC & \DD
\end{bmatrix},
\]
we denote the set of leading indices with $F$ and the rest of indices are denoted with $C$, i.e., $\MM_{FF} = \AA$, and $\MM_{CC} = \DD$. 
Then the Schur complement of $\MM$ with respect to indices $C$ is $\sc(\MM,C) = \DD - \CC \AA^{-1} \BB$. 
Then if $\AA$ and $\SS := \sc(\MM,C)$ are invertible, $\MM$ is invertible and its inverse is the following.
\begin{align}
\label{eq:sc-inversion}
\MM^{-1} = \begin{bmatrix}
    \AA^{-1} + \AA^{-1} \BB \SS^{-1} \CC \AA^{-1} & -\AA^{-1} \BB \SS^{-1} \\
    - \SS^{-1} \CC \AA^{-1} & \SS^{-1}
\end{bmatrix}.
\end{align}

\paragraph{Floating-point numbers.} A floating point number in base $t$ is stored using two integer scalars $a$ and $b$ as $a \cdot t^{b}$. 
Scalar $a$ is the significand and scalar $b$ is the exponent. 
The most usual choice of the base on real computers is $t=2$. 
The number of bits of $a$ determines the bits of precision and the number of bits of $b$ determines how large or small our numbers can be. 
Let $L$ be the number of bits that we allow for $a$ and $b$, then for any number $c \in [t^{- 2^{L}+1},(2^{L} - 1) \cdot t^{2^{L} - 1}]$, there is a number $d$ in floating-point numbers with $L$ bits such that $d \leq c \leq (1 + t^{- 2^{L-1}}) d$. 
Therefore if our numbers are not too large or too small, we can approximate them using $O(\log(1/\epsilon))$ bits to $e^{\epsilon}$ multiplicative error. 
Also, note that we can use floating-point numbers to show \emph{zero} exactly. 
Using $\Otil(\log(1/\epsilon))$ bit operations, by fast Fourier transform (FFT), we can perform the multiplication of two numbers to $e^{\epsilon}$ accuracy. 
The sign of the result will be correct. Moreover with $O(\log(1/\epsilon))$ bit operations, we can perform an addition of two nonnegative numbers (or two nonpositive numbers) with a multiplicative error of $e^{\epsilon}$. 
However, if we add a positive number to a negative number, the error will be additive and depends on $\epsilon$ --- see Kahan's summation algorithm \cite{K65}. 
That is why we do not perform the addition of positive and negative numbers in our algorithms in this paper.

\section{Matrix Associated with Escape Probability}
\label{sec:linear-system}
In this section, we characterize the linear system corresponding to the computation of escape probabilities and discuss the invertibility of RDDL matrices arising from such linear systems. In \Cref{sec:sc-view-for-escape-prob}, we observe that the linear system can be transformed into a system with two fewer equations and variables by looking at the inverse via Schur complement.


\begin{lemma}
\label{lem:rddl-invertible}
Let $\MM \in \R^{n\times n}$ be an RDDL matrix for which the corresponding graph is connected to the dummy vertex. 
Then $\MM$ is invertible.
\end{lemma}
\begin{proof}
    Let $\DD \in \R^{n\times n}$ be a diagonal matrix where for $i\in[n]$, $\DD_{ii} = \MM_{ii}$. 
    First note that $\MM$ does not have a row of all zeros since otherwise the corresponding vertex is an isolated vertex and the graph corresponding to $\MM$ is not connected to the dummy vertex. 
    Therefore all $\DD_{ii}$'s are nonzero and $\DD$ is invertible.
    Let $\NN = \II - \DD^{-1} \MM$ and $k\in\N$. 
    Then $(\NN^{k})_{ij}$ is the probability that a random walk of size $k$ that started at vertex $i$ ends at vertex $j$ without hitting the dummy vertex. 
    Since by assumption, for all vertices, there is a path of positive edge weights to the dummy vertex, as $k\to \infty$, $(\NN^{k})_{ij} \to 0$. 
    Therefore the spectral radius of $\NN$ is strictly less than one. 
    Thus $\DD^{-1} \MM = \II - \NN$ does not have a $0$ eigenvalue and is invertible. 
    Therefore, $\MM$ is also invertible.
\end{proof}

We now characterize the linear system that needs to be solved to compute escape probabilities in a graph. In the following, matrix $\AA$ is the matrix corresponding to the Markov chain (random walk) in the graph (i.e., $\AA_{ij}$ is the probability of going from vertex $i$ to vertex $j$) with the only difference that the rows corresponding to vertices $t$ and $p$ are zeroed out. This is to prevent the random walk to exit $t$ or $p$ when it arrives at one of them.

\begin{lemma}
\label{lem:escape-prob-matrix}
Let $\AA$ be the matrix obtained by zeroing out the rows corresponding to vertices $t$ and $p$ in the random walk associated with (directed or undirected and weighted or unweighted) graph $G=(V,E)$. If for all $v \in V$, there is a path from $v$ to $t$ or $p$ then $\II - \AA$ is invertible and $\P(s,t,p) = (\II-\AA)^{-1}_{st}$ and $\P(s,p,t) = (\II-\AA)^{-1}_{sp}$.
\end{lemma}

\begin{proof} 

Let $\DD$ be the principal submatrix corresponding to $t$ and $p$ in $\II-\AA$ and $\MM$ be the principal submatrix corresponding to the rest of the vertices. Since the rows corresponding to $t$ and $p$ in $\AA$ are zero, it is easy to see that 
\[
\DD = \begin{bmatrix}
    1 & 0 \\
    0 & 1
\end{bmatrix}.
\]
and $\sc(\II-\AA, [n]\setminus\{t,p\}) = \MM$. Obviously, $\DD$ is invertible. Moreover, by assumption, $\MM$ is connected to the dummy vertex. The dummy vertex of $\MM$ is essentially the contraction of $t$ and $p$. Therefore by \Cref{lem:rddl-invertible}, $\MM$ is also invertible. Therefore since $\DD$ and the Schur complement are invertible, $\II-\AA$ is invertible.

Then, we prove the theorem by induction. First, note that for each vertex $i \in [n] \setminus \{t,p\}$, the corresponding row of $\AA$ is a probability distribution over the endpoints of random walks of size one that starts at vertex $i$. Now consider entry $i,j$ of $\AA^2$. We have 
\begin{align*}
\left(\AA^2\right)_{ij} = \sum_{k=1}^n \AA_{ik} \AA_{kj}. 
\end{align*}
Note that $\AA_{ik}$ is the probability of going from vertex $i$ to vertex $k$ in the first step of the random walk and $\AA_{kj}$ is the probability of going from $k$ to $j$ in the second step. Therefore $(\AA^2)_{ij}$ is the probability of ending up at vertex $j$ with a random walk of size $2$ that starts at $i$. Continuing this argument inductively, we can argue that $(\AA^k)_{ij}$ denotes the probability of ending up at vertex $j$ with a random walk of size $k$ that starts at $i$.

Since $\II-\AA$ is invertible, $\left(\II-\AA\right)^{-1} = \II + \AA + \AA^2 + \cdots$. Therefore
\begin{align*}
\left(\II-\AA\right)^{-1}_{st} = \II_{st} + \AA_{st} + \left(\AA^2\right)_{st} + \left(\AA^3\right)_{st} + \cdots
\end{align*}
Consider a random walk $v_1,v_2,v_3,\ldots,v_{k+1}$, where $v_1=s$ and $v_{k+1} = t$. Note that
\begin{align*}
\left(\AA^k\right)_{st} & = \sum_{j_{k-1}=1}^n \left(\AA^{k-1}\right)_{s j_{k-1}} \AA_{j_{k-1} t} = \sum_{j_{k-1}=1}^n \sum_{j_{k-2}=1}^n \left(\AA^{k-2}\right)_{s j_{k-2}} \AA_{j_{k-1} j_{k-2}} \AA_{j_{k-1} t}
\\ & =
\sum_{j_{k-1}=1}^n \sum_{j_{k-2}=1}^n \cdots \sum_{j_{1}=1}^n \AA_{sj_1} \cdots \AA_{j_{k-3} j_{k-2}} \AA_{j_{k-1} j_{k-2}} \AA_{j_{k-1} t}
\end{align*}
Now, note that $\AA_{v_1 v_2} \AA_{v_2 v_3} \cdots \AA_{v_k v_{k+1}}$ is a summand of the above summation. Therefore the probability of random walk $v_1,v_2,v_3,\ldots,v_{k+1}$ is counted in $\AA^k$. Now we argue that it is counted in only one of $\AA^i$'s. Suppose for $i\in[k]\setminus\{1\}$, $v_i$ is equal to $t$ or $p$. In this case since $\AA_{tj}=\AA_{pj}=0$ for any $j\in[n]$, then $\AA_{v_i v_{i+1}}=0$ and therefore the probability of this random walk is equal to zero. In other words, any random walk that visits either $t$ or $p$ is only counted in the $\AA^i$ corresponding to the first visit to $t$ or $p$. Another view on this is that when a random walk visits $t$ or $p$, it stays there with probability of one for all the rest of the steps. This concludes the proof. 
\end{proof}

By \Cref{lem:escape-prob-matrix}, one can see that it is enough to solve the linear system $\left(\II - \AA\right) \xx = \boldsymbol{e}_{t}$ to compute $\P(s,t,p)$ and the solution gives the probabilities for all $s\in V$. Note that the solution to this linear system is column $t$ of the inverse of $\II-\AA$. In the next section, we identify another linear system that gives the escape probabilities.

\subsection{Schur Complement View}
\label{sec:sc-view-for-escape-prob}

Here we take a more careful look at the columns corresponding to $t$ and $p$ in the inverse of $\II-\AA$ with the help of Schur complement.
Without loss of generality suppose $t=n-1$ and $p=n$. Then matrix $\II-\AA$ is as the following.
\begin{align*}
\II - \AA = \begin{bmatrix}
\MM & \BB \\
\CC & \DD
\end{bmatrix},
\end{align*}
where 
\begin{align*}
\DD = \begin{bmatrix}
    1 & 0 \\
    0 & 1
\end{bmatrix},
\end{align*}
and $\CC$ is a $2$-by-$(n-2)$ matrix of all zeros. Then by \eqref{eq:sc-inversion}, we have
\begin{align*}
\left(\II-\AA\right)^{-1}_{:,n-1:n} = \begin{bmatrix}
- \left(\MM - \BB \DD^{-1} \CC\right)^{-1} \BB \DD^{-1}
\\
\DD^{-1} + \DD^{-1} \CC \left(\MM - \BB \DD^{-1} \CC\right)^{-1} \BB\DD^{-1}
\end{bmatrix}
\end{align*}
Therefore noting that $\DD$ is the identity matrix and $\CC$ is an all zero matrix, we have
\begin{align}
\label{eq:escape-formula}
\left(\II-\AA\right)^{-1}_{:,(n-1):n} = 
\begin{bmatrix}
- \MM^{-1} \BB
\\
\DD
\end{bmatrix}
\end{align}
Therefore computing $- \MM^{-1} \BB$, which corresponds to solving two linear systems, gives both escape probabilities $\P(s,t,p)$ and $\P(s,p,t)$ for $s\in V \setminus \{t,p\}$. Also, note that $\P(t,t,p) = \P(p,p,t) = 1$, and $\P(p,t,p) = \P(t,p,t) = 0$.

\section{Repeated Squaring}
\label{sec:repeated-square}

As mentioned in the proof of \Cref{lem:escape-prob-matrix}, if $\II - \AA$ is invertible, the inverse can be computed by the means of the power series $\II + \AA + \AA^{2} + \cdots$. 
However, to design an algorithm from this, we need to cut the power series 
 and take the summation for a finite number of terms in the power series. 
 In other words, we need to output $\II + \AA + \AA^{2} + \cdots + \AA^{k}$ (for some $k\in\N$) as an approximate inverse. 
 Then such an approximate inverse can be computed efficiently by utilizing a repeated squaring approach. 
 
 In this section, we give bounds for the required number of terms to produce appropriate approximations. 
 A disadvantage of this approach compared to the recursion approach that we discuss in the next section is that the running time would have a dependence on the logarithm of hitting times to $t$ and $p$.

 We start by formally defining the set of random walks that hit a certain vertex only in their last step and use that to define hitting times. We then show that such hitting times characterize how fast the terms in the power series decay. Finally, we use that to bound the number of terms needed to approximate the inverse well.



\begin{definition}
\label{def:hitting-time}
    For a (directed or undirected and weighted or unweighted) graph $G=(V,E)$ and $s,t\in V$, we define the hitting time of $s$ to $t$, as the expected number of steps for a random walk starting from $s$ to reach $t$ for the first time. 
    More formally, let 
    \begin{align}
        W_{st}=\{w=(v_1,\ldots,v_k): k\in\N, \forall i\in [k], \forall j\in[k-1], v_i\in V, (v_j,v_{j+1}) \in E,v_j\neq t, v_1=s, v_k = t\},
    \end{align}
    be the set of all possible walks from $s$ to $t$ that visit $t$ only once. We set the probability of the random walk $w=(v_1,\ldots,v_k)$ equal to $\pr(w):= \pr(v_1,v_2) \pr(v_2,v_3) \cdots \pr(v_{k-1},v_k)$, where $\pr(u,v)$ is the probability of going from vertex $u$ to $v$ in one step. 
    We also denote the size of $w$ with $|w|$ which is the number of edges traversed in $w$, i.e., for $w=(v_1,\ldots,v_k)$, $|w|=k-1$. Then the hitting time of $s$ to $t$ is
    \begin{align}
        H(s,t):=\sum_{w\in W_{st}} \pr(w) \cdot |w|.
    \end{align}
    If there is no path from $s$ to $t$, i.e., $W_{st}=\emptyset$, then $H(s,t)=\infty$.
\end{definition}

The way \Cref{def:hitting-time} defines the hitting time, it only works for one vertex. 
To define the hitting time for hitting either $t$ or $p$, we can consider the graph $\overline{G}$ obtained from $G$ by contracting $t$ and $p$. 
We denote the vertex corresponding to the contraction of $t$ and $p$ with $q$.
Then the probability of going from $s$ to $q$ (in one step) in $\overline{G}$ denoted by $\overline{\pr}(s,q)$ is equal to $\pr(s,t)+\pr(s,p)$. 
Then the hitting time of $s$ to $q$ in $\overline{G}$ can be defined in the same manner as \Cref{def:hitting-time} and we denote it with $\overline{H}(s,q)$. 
The following lemma bounds the decay of the terms of the power series using this hitting time.

\begin{lemma}
\label{lemma:2t}
For a (directed or undirected and weighted or unweighted) graph $G=(V,E)$ and $s,t,p\in V$, let $\overline{G} =(\overline{V},\overline{E})$ be the graph obtained by contracting (identifying) $t$ and $p$. 
Let $q$ be the vertex corresponding to $t$ and $p$ in $\overline{G}$.

Let $\AA$ be the random walk matrix associated with graph $G=(V,E)$ in which the rows corresponding to $t$ and $p$ are zeroed out.
Then for any $h$ and any $k \ge 2h$ such that
\[
h
\geq
1+\max_{s \in \overline{V}} \overline{H}\left(s,q)\right.
\]
we have 
\[
\left\| \left(\AA^{k}\right)_{u:} \right\|_1
\leq
\frac{1}{2}
\qquad
\forall u \in V . 
\]
That is, any row of $\AA^{k}$ sums up to less than $0.5$, and thus $\|\AA^k\|_\infty \le 0.5.$
\end{lemma}

\begin{proof}
    By definition, we have $\overline{H}(s,q) \leq h$, for any $s\in \overline{V}$. Let $w$ be a random walk from $s$ to $q$ in $\overline{G}$. We have $\E(|w|)< h$. Therefore, by the Markov inequality, $\P[|w|\leq k] \geq \P[|w|\leq 2h] \geq \frac{1}{2}$. Note that entry $u \in \overline{V}\setminus \{q\}$ in row $s$ of $\AA^{k}$ denotes the probability that a random walk of size $k$ in $\overline{G}$ that starts at $s$ ends at $u$. We need to clarify here that a random walk that reaches $t$ (or $p$) in $2h$ steps for the first time stays at $t$ (or $p$). However, the probability of a random walk that reaches $t$ (or $p$) in the $k$'th step and stays there does not get added to $\AA^{k+1}$ (or any of the terms after that) because row $t$ (and $p$) of $\AA$ are zero and therefore the last term in the probability of the random walk is zero. This argument clarifies that we do not ``double count'' the probability of random walks in this calculation. 
    
    Note that $\P[|w|\leq k] \geq \frac{1}{2}$ implies that any random walk of size $k$ ends up at $q$ with probability at least $0.5$ since the random walk does not exit $q$ when it reaches it. Therefore the total probability for all the other vertices in $\overline{V}\setminus \{q\}$ is at most $0.5$. Therefore $\| (\AA^{k})_{u:} \|_1 \leq \frac{1}{2}$.
\end{proof}

Here, we recall the property of matrix infinity norm. For $k \ge 2h$, multiplication by $\AA^{k}$ makes any vector smaller. 

\begin{corollary} \label{coro:infinity-norm}
    For any vector $\vv$,  we have $\| \AA^{k} \vv \|_{\infty} \leq  \| \AA^{k} \|_{\infty} \| \vv \|_\infty \le \frac{1}{2}\|\vv\|_{\infty}$. 
\end{corollary}



We need the following lemma for the maximum hitting time for unweighted undirected graphs to bound the number of terms required in the power series.

\begin{lemma}[\hspace{-0.05pt}\cite{F17}]
\label{lemma:m2}
    For an undirected unweighted graph $G$, the maximum hitting time is at most $m^2$.
\end{lemma}

We are now prepared to prove the main result of this section which gives an algorithm for computing escape probabilities in graphs with bounded polynomial weights.

\begin{theorem}
\label{lemma:2tg}
    Let $G=(V,E)$ be a undirected graph with integer weights in $[1, n^c]$, for $c\in\N$. Given $t,p\in V$, we can compute the escape probability $\P(s,t,p)$ for all $s\in V$ in $\Otil(n^3 c \log(c) \log\frac{1}{\epsilon})$ bit operations within a $e^{\epsilon}$ multiplicative factor.
\end{theorem}

\begin{proof}
Let $\AA$ and $h$ be as defined in \Cref{lemma:2t}. Without loss of generality, we assume the graph is connected. There are three cases that need to be handled first:
\begin{itemize}
    \item All paths between $s$ and $p$ contain $t$: This indicates $\P(s,t,p)=1$. This can be identified by removing $t$ and checking if $s$ and $p$ are disconnected.
    \item All paths between $s$ and $t$ contain $p$: This indicates $\P(s,t,p)=0$. This can be identified by removing $p$ and checking if $s$ and $t$ are disconnected.
    \item There is no path from $s$ to $t$ and $p$. In this case, the escape probability is undefined.
\end{itemize}
The above cases can be handled by running standard search algorithms on the graph (e.g., breadth-first search) and removing the vertices $s$ with such escape probabilities. The running time of such a procedure is smaller than the running time stated in the theorem.
Ruling out these cases, we can assume all vertices $s\in V\setminus\{t,p\}$ have paths with positive edge weights to both $t$ and $p$, $\II-\AA$ is invertible, and $\sum_{i=1}^{\infty} \AA^{i}$ is convergent.

Let $k = (2h+1)(2n(c+1)\lceil \log \frac{(2h+1)n}{\epsilon} \rceil + 1) - 1$, where $h$ is defined as in \Cref{lemma:2t}.
We show
\begin{align*}
\left( \II - \AA \right)^{-1}
\approxbar_{\epsilon}
\sum_{i = 0}^{k}
\AA^{i}.
\end{align*}
First note that 
    \begin{align*}
        \left(\II-\AA\right)^{-1} = \sum_{i = 0}^{\infty} \AA^i.
    \end{align*} 
    Let $\BB=\sum_{i = k +1}^{\infty} \AA^i$. Since $\AA^i \geq 0$, for all $i\geq 0$, $\BB\geq 0$. Therefore, we have
    \begin{align*}
        \left(\II-\AA\right)^{-1} \geq \sum_{i = 0}^{k}\AA^i
    \end{align*}.
    Moreover, we have
    \begin{align}
        \left(\II-\AA\right)^{-1} = (\sum_{j=0}^{\infty} \AA^{(2h+1)j}) (\sum_{i = 0}^{2h}\AA^i).
    \end{align}
    Note that by \Cref{coro:infinity-norm} and induction,
    \begin{align}
        \norm{ \AA^{(2h+1)j} \vv }_{\infty} & =
        \norm{ \AA^{2h+1}(\AA^{(2h+1)(j-1)} \vv)}_{\infty} \nonumber
        \\ & \leq
        \frac{1}{2}  \norm{\AA^{(2h+1)(j-1)} \vv}_{\infty} \nonumber
        \\ & \leq
        \frac{1}{2^j}  \norm{\vv}_{\infty} \label{eq:things_get_small}
    \end{align}

    Note that the $\ell_1$ norm of each column of $\AA^{i}$ is bounded by $n$ because all the entries are in $[0,1]$. Therefore, again because all entries are between $[0,1]$ (i.e., they are non-negative), the $\ell_1$ norm of each column of $\sum_{i = 0}^{2h}\AA^i$ is bounded by $(2h+1)n$. Therefore by setting $\vv$ to be a column of $\sum_{i = 0}^{2h}\AA^i$ and using \eqref{eq:things_get_small} and triangle inequality, we have that 
    \begin{align}
        \norm{(\sum_{j=2n(c+1)\lceil\log((2h+1)n/\epsilon)\rceil+1}^{\infty} \AA^{(2h+1)j}) \vv}_{\infty} 
        & \leq \sum_{j=2n(c+1)\lceil\log((2h+1)n/\epsilon)\rceil+1}^{\infty} \norm{ \AA^{(2h+1)j} \vv }_{\infty} \nonumber
        \\ & \leq 
    \sum_{j=2n(c+1)\lceil\log((2h+1)n/\epsilon) \rceil+1}^{\infty} \frac{1}{2^j}\norm{ \vv }_{\infty} \nonumber
    \\ & \leq 
    \frac{1}{2^{2n(c+1)\lceil\log((2h+1)n/\epsilon) \rceil}} \norm{ \vv }_{\infty} \nonumber
    \\ & \leq
    \frac{\epsilon}{n^{n(c+1)} \cdot (2h+1) n} \norm{ \vv }_{\infty} \nonumber
    \\ & \leq
    \frac{\epsilon}{n^{n(c+1)}}. \label{eq:very-small}
    \end{align}

    Note that any escape probability is at least $n^{-n(c+1)}$ since any $\pr(x,y)$ is at least $1/n^{c+1}$. 
    Therefore by \eqref{eq:very-small}, not considering the terms of the power series with an exponent larger than $(2h+1)(2n(c+1)\lceil\log((2h+1)n/\epsilon)\rceil+1)$ can only cause a multiplicative error of $O(e^{\epsilon})$ in the computation of the escape probability. Therefore our algorithm is to compute the matrix
    \[
    \XX = \sum_{i=0}^{k-1} \AA^{i},
    \]
    and return $\XX_{:t}$,
    where $k$ is the smallest power of two larger than $(2h+1)(2n(c+1)\lceil \log \frac{(2h+1)n}{\epsilon} \rceil + 1)$. To compute $\XX$ we use the recursive approach of repeated squaring. Namely, let $k=2^r$. Then

\begin{align*}
\XX
= \sum_{i=0}^{2^r-1} \AA^i
= \left(\sum_{i=0}^{2^{r-1}-1} \AA^i\right)
\left(\II + \AA^{2^{r-1}}\right)
=
\cdots
=
\left(\II+\AA\right)
\left(\II+\AA^2\right)
\left(\II+\AA^4\right)
\cdots
\left(\II+\AA^{2^{r-1}}\right).
\end{align*}
Therefore, $\XX$ can be computed with $O(\log(k))$ matrix multiplications.

In addition, all of $\AA^{2^{\widehat{r}}}$ ($\widehat{r}\in[r-1]$) can be computed with $\log(k)$ matrix multiplications, by observing that $\AA^{2^{\widehat{r}}} = \AA^{2^{\widehat{r}-1}}\AA^{2^{\widehat{r}-1}}$.
Now, suppose each floating-point operation we perform incurs an error of $e^{\widehat{\epsilon}}$. Therefore if we use an algorithm with $O(n^3)$ number of arithmetic operations to compute the matrix multiplications, then we incur a multiplicative error of at most $e^{2n\widehat{\epsilon}}$. Then the error of computation of $\AA^{2^{\widehat{r}}}$ is at most $e^{2n r\widehat{\epsilon}}$ and the total error of computing $\XX$ is at most $e^{(2n r+1)r\widehat{\epsilon}}$. Setting $\widehat{\epsilon} = \frac{\epsilon}{(2nr+1)r}$, the total error will be at most $e^{\epsilon}$. Note that to achieve this we need to work with floating-point numbers with $O(\log\frac{(2n\log k)\log k}{\epsilon})$ bits. This shows that our algorithm requires only
\[
\Otil(n^3 (\log k)(\log\frac{(2n\log k)\log k}{\epsilon}))
\]
bit operations. To finish the proof, we need to bound $k$. To bound $k$, we need to bound $h$. By \Cref{lemma:m2}, the maximum hitting time for undirected unweighted graphs is $m^2$. Note that this gives a bound for our hitting time to either of $t$ or $p$ as well in the unweighted case since after handling the pathological cases at the beginning of the proof, we assumed that for every vertex $s\in V \setminus \{t,p\}$, there exist paths to both $t$ and $p$. Moreover introducing polynomial weights in the range of $[1,n^c]$ can only increase the hitting time by a factor of $n^c$. Therefore $h=O(n^{c+4})$. Then $k = \Otil(n^{c+5} c \log(\frac{n^{c+5}}{\epsilon}))$. Therefore $\log(k) = \Otil (c \cdot \log c)$. Therefore
\[
n^3 (\log k)(\log\frac{(2n\log k)\log k}{\epsilon}) = \Otil(n^3 \cdot c \cdot \log(c) \cdot \log(\frac{1}{\epsilon})).
\]
\end{proof}

\section{Floating Point Edge Weights}
\label{sec:recursive-SC}

In this section, we prove the main results of the paper. We start by characterizing the entries of the inverse of RDDL matrices. To do so, we need the following theorem, which is a corollary of the result of \cite{C82} that proves a more general matrix-tree theorem.

\begin{theorem}[Matrix-Tree Theorem \cite{C82}]
\label{thm:matrix-tree}
    Let $\MM \in \R^{n\times n}$ be an RDDL matrix such that for all $i\in [n]$, $\sum_{j\in [n]} \MM_{ij} = 0$. Let $i \in [n]$ and $\NN$ be the matrix obtained by removing row $i$ and column $i$ from $\MM$. Then $\det(\NN) = \sum_{F \in \mathcal{F}} \prod_{e\in F} w_e$, where $\mathcal{F}$ is the set of all spanning trees oriented towards vertex $i$ in the graph corresponding to matrix $\MM$. $\prod_{e\in F} w_e$ is the product of the weights of the edges of spanning tree $F$.
\end{theorem}

We colloquially refer to the summation $\sum_{F \in \mathcal{F}} \prod_{e\in F} w_e$ as the \emph{weighted number of spanning trees oriented toward vertex $i$}. 
This is because, in the case of unweighted graphs, the summation is equal to the number of spanning trees. 
The following is a direct consequence of the matrix-tree theorem.

\begin{lemma}
\label{lem:det-of-rddl}
    Let $\MM \in \R^{n\times n}$ be an RDDL matrix. Then $\det(\MM)$ is equal to the weighted number of spanning trees oriented towards a vertex in a graph.
\end{lemma}
\begin{proof}
    We add a dummy vertex (i.e., a row and column) to the graph associated with $\MM$ to obtain matrix $\YY$. 
    The principal submatrix of $\YY$ associated with indices in $[n]$ is equal to $\MM$. 
    Row $n+1$ of $\YY$ is all zeros and entry $i \in [n]$ of column $n+1$ of $\YY$ is $-\sum_{j=1}^n \MM_{ij}$. 
    Then by \Cref{thm:matrix-tree} and removing the dummy vertex, it immediately follows that $\det(\MM)$ is the weighted number of spanning trees oriented towards vertex $n+1$ in the graph corresponding to $\YY$.
\end{proof}

    Note that in \Cref{lem:det-of-rddl}, if $\MM$ is not connected to the dummy vertex, then there is no spanning tree oriented towards the dummy vertex and $\det(\MM)=0$, i.e., $\MM$ is singular.

\begin{lemma}
\label{lemma:entries-of-inverse}
Let $\MM$ be an RDDL matrix. Then if $\MM$ is invertible, each entry of $\MM^{-1}$ is the ratio of the weighted number of spanning trees in two graphs (oriented towards some vertex of each graph), both with positive weights obtained from entries of $\MM$. 
\end{lemma}
\begin{proof}

    Note that the solution to $\MM \xx = \ee^{(j)}$ gives the $j$th column of $\MM^{-1}$, where $\ee^{(j)}$ is the $j$th standard basis vector. Let $\MM^{(i,j)}$ be the matrix obtained from $\MM$ by replacing the $i$th column with $\ee^{(j)}$. Then by Cramer's rule 
    \begin{align}
    \label{eq:entries-of-inverse}
        \MM^{-1}_{ij} = \frac{\det\left(\MM^{(i,j)}\right)}{\det\left(\MM\right)}.
    \end{align}

By \Cref{lem:det-of-rddl}, $\det(\MM)$ is the weighted number of spanning trees oriented towards a vertex in a graph.
    For $\det(\MM^{(i,j)})$ note that the determinant is equal to
    \begin{align*}
        \det\left(\MM^{(i,j)}\right) = \sum_{\sigma\in S_n} \sign(\sigma) \prod_{k=1}^n \MM^{(i,j)}_{k,\sigma(k)}.
    \end{align*}
    Note that if $\sigma(k)=j$ and $k\neq i$, then $\prod_{k=1}^n \MM^{(i,j)}_{k,\sigma(k)} = 0$. This is because we have a term $\MM_{i,\ell}$ in the product with $\ell\neq j$ and any such term is equal to zero.
    Therefore, if $\widehat{\MM}^{(i,j)}$ is the matrix obtained from $\MM^{(i,j)}$ by removing the $i$th row and $j$th column. Then $\det(\MM^{(i,j)}) = (-1)^{i+j}\det(\widehat{\MM}^{(i,j)})$. Note that $\MM^{(i,j)}$ is an $(n-1)$-by-$(n-1)$ matrix.

    There are three cases based on how $i$ and $j$ compare.
    
    \textbf{Case 1:} $i = j$.
    $\widehat{\MM}^{(i,i)}$ is a principal submatrix of $\MM$ and therefore it is invertible and RDDL. Therefore \Cref{lem:det-of-rddl} immediately resolves this case.

    \textbf{Case 2:} $i > j$.
    We first push rows $j+1,\ldots,i-1$ up and move row $j$ to become row $i-1$ to obtain the matrix $\overline{\MM}^{(i,j)}$. This requires first switching rows $j$ and $j+1$, then switching rows $j+1$ and $j+2$, and so on. 
    Since we have $i-1-j$ row switches, then
    \[
    \det \left(\widehat{\MM}^{(i,j)} \right)
    = (-1)^{i-1-j} \det \left(\overline{\MM}^{(i,j)} \right).
    \]
    Therefore, 
    \[
    \det\left(\MM^{(i,j)} \right) = - \det\left(\overline{\MM}^{(i,j)} \right).
    \]
    Note that all the diagonal entries of $\overline{\MM}^{(i,j)}$ are positive except $\overline{\MM}^{(i,j)}_{i-1,i-1}$.
    We now construct $\widetilde{\MM}^{(i,j)}$ from $\overline{\MM}^{(i,j)}$ as the following in two steps:
    \begin{enumerate}
    \item We first take the sum of all columns of $\overline{\MM}^{(i,j)}$ except column $i-1$ and add it to column $i-1$;
    \item We negate the resulting column $i-1$.
    \end{enumerate}
    First, note that the first operation does not change the determinant, and the second operation changes the sign of the determinant.
    Therefore
    \begin{align}
    \label{eq:m-tilde-ij-det}
    \det \left(\MM^{(i,j)} \right) = \det \left(\widetilde{\MM}^{(i,j)} \right).
    \end{align}
    Since for row $i-1$ of $\overline{\MM}^{(i,j)}$, all the entries are non-positive,
    \[
    \widetilde{\MM}^{(i,j)}_{i-1,i-1} \geq 0.
    \]
    Moreover, by construction 
    \[
    \sum_{k\in[n-1]\setminus\{i-1\}} |\widetilde{\MM}^{(i,j)}_{i-1,k}| \leq |\widetilde{\MM}^{(i,j)}_{i-1,i-1}|.
    \]
    For $\ell\neq i-1$, 
    \[
    \widetilde{\MM}^{(i,j)}_{\ell,i-1} = - \overline{\MM}^{(i,j)}_{\ell,\ell}  - \sum_{k\in[n-1]\setminus \{\ell\}} \overline{\MM}^{(i,j)}_{\ell,k} 
    \]
    Therefore since for $k\in [n-1]\setminus \{\ell\}$, $\overline{\MM}^{(i,j)}_{\ell,k} \leq 0$, and 
    \[
    \sum_{k\in[n-1]\setminus \{\ell\}} |\overline{\MM}^{(i,j)}_{\ell,k}| \leq \overline{\MM}^{(i,j)}_{\ell,\ell},
    \]
    we have,
    \[
    \sum_{k\in[n-1]\setminus \{\ell\}} \widetilde{\MM}^{(i,j)}_{\ell,k} = \overline{\MM}^{(i,j)}_{\ell,\ell} - \left( \sum_{k\in[n-1]\setminus \{\ell\}} |\overline{\MM}^{(i,j)}_{\ell,k}| \right) + \left( \sum_{k\in[n-1]\setminus \{\ell\}} |\overline{\MM}^{(i,j)}_{\ell,k}| \right) = \widetilde{\MM}^{(i,j)}_{\ell,\ell}.
    \]
    Therefore $\widetilde{\MM}^{(i,j)}$ is RDDL. Now applying \Cref{lem:det-of-rddl} to $\widetilde{\MM}^{(i,j)}$ proves the result for this case.


    \paragraph{Case 3:} $i < j$.
    This is very similar to the $i > j$ case above and therefore we omit the proof of this. The only difference between this case and Case 2 is that we need a different row switching to obtain the matrix $\overline{\MM}^{(i,j)}$. More specifically, we need to push rows $i,\ldots, j-2$ of $\widehat{\MM}^{(i,j)}$ down and move row $j-1$ to row $i$. Moreover, we need to obtain $\widetilde{\MM}^{(i,j)}$ from $\overline{\MM}^{(i,j)}$ by changing column $i$ instead of column $i-1$. 
\end{proof}

\begin{lemma}
\label{lem:ApproxInvert}
Let $\MM,\NN \in \R^{n \times n}$ be invertible RDDL matrices such that
\[
\MM \approxbar_{\epsilon} \NN,
\]
and for all $i\in[n]$,
\[
\sum_{j\in \left[n\right]} \MM_{ij}
\approx_{\epsilon}
\sum_{j\in \left[n\right]} \NN_{ij}.
\]
Then $\MM^{-1} \approxbar_{2\epsilon n} \NN^{-1}$.
\end{lemma}

Note that just the first condition in \Cref{lem:ApproxInvert} alone is not sufficient for the inverses to be close to each other. For example, consider matrices
\[
\left[
\begin{matrix}
1 & -1\\
-1 & 1
\end{matrix}
\right],
\text{ and }
\left[
\begin{matrix}
1 + \epsilon &  -1\\
-1 & 1
\end{matrix}
\right].
\]
The first one is singular while the second one is invertible.

\begin{proof}
    We define $\widetilde{\NN}^{(i,j)}$ similar to the definition of $\widetilde{\MM}^{(i,j)}$ in the proof of \Cref{lemma:entries-of-inverse}.
    By \eqref{eq:entries-of-inverse} and \eqref{eq:m-tilde-ij-det}, we have
    \[
    \MM^{-1}_{ij} = \frac{\det\left(\widetilde{\MM}^{(i,j)}\right)}{\det\left(\MM\right)}
    \]
    and
    \[
    \NN^{-1}_{ij} = \frac{\det\left(\widetilde{\NN}^{(i,j)}\right)}{\det\left(\NN\right)}
    \]
    First note that since $\det\left(\MM\right)$ and $\det\left(\NN\right)$ are the sum of the product of the weight of the edges of spanning trees and the edge weights in graphs corresponding to $\MM$ and $\NN$ are within a factor of $e^{\epsilon}$ from each other, $\det\left(\NN\right) \approxbar_{\epsilon n} \det\left(\MM\right)$. A similar argument applies to Case 1 in the proof of \Cref{lemma:entries-of-inverse}, which gives 
    \[
    \det\left(\widehat{\MM}^{(i,i)}\right) \approxbar_{\epsilon n} \det\left(\widehat{\NN}^{(i,i)}\right).
    \]
    
    Now consider Case 2 in the proof of \Cref{lemma:entries-of-inverse}. Trivially the edge weights in $\widetilde{\MM}^{(i,j)}$ and $\widetilde{\NN}^{(i,j)}$ outside column $i-1$ are within a factor of $e^{\epsilon}$ of each other. For edge weights in column $i-1$, note that for $\ell\neq i-1$,
    \[
    \widetilde{\MM}^{(i,j)}_{\ell,i-1} = - \overline{\MM}^{(i,j)}_{\ell,\ell}  - \sum_{k\in[n-1]\setminus \{\ell\}} \overline{\MM}^{(i,j)}_{\ell,k} 
    \]
    and 
    \[
    \widetilde{\NN}^{(i,j)}_{\ell,i-1} = - \overline{\NN}^{(i,j)}_{\ell,\ell}  - \sum_{k\in[n-1]\setminus \{\ell\}} \overline{\NN}^{(i,j)}_{\ell,k}. 
    \]
    Let $\ell'$ be the index of the row in $\MM$ corresponding to row $\ell$ in $\widetilde{\MM}^{(i,j)}$. Then
    \[
    \widetilde{\MM}^{(i,j)}_{\ell,i-1}
    =
    -\sum_{k \in [n] \setminus\{j\}} \MM_{\ell',k}
    =
    \left(-\sum_{k \in [n]} \MM_{\ell',k}\right) + \MM_{\ell',j}.
    \]
    Similarly, we have 
    \[
    \widetilde{\NN}^{(i,j)}_{\ell,i-1}
    =
    \left(-\sum_{k \in [n]} \NN_{\ell',k}\right) + \NN_{\ell',j}.
    \]
    Since $\MM$ and $\NN$ are RDDL matrices, $(-\sum_{k \in [n]} \MM_{\ell',k})$, $(-\sum_{k \in [n]} \NN_{\ell',k})$, $\MM_{\ell',j}$, $\NN_{\ell',j}$ are nonpositive numbers. Therefore, since $\sum_{k\in \left[n\right]} \MM_{\ell',k}
\approx_{\epsilon}
\sum_{k\in \left[n\right]} \NN_{\ell',k}$ and $\MM_{\ell',j} \approx_{\epsilon} \NN_{\ell',j}$, we have 
\[
\widetilde{\MM}^{(i,j)}_{\ell,i-1} \approx_{\epsilon n} \widetilde{\NN}^{(i,j)}_{\ell,i-1}.
\]
    Thus, 
    \[
    \det\left(\widetilde{\NN}^{(i,j)}\right) \approxbar_{\epsilon n} \det\left(\widetilde{\MM}^{(i,j)}\right).
    \]
    Therefore, we have
    \[
    \MM^{-1}_{ij} \approxbar_{2\epsilon n} \NN^{-1}_{ij}.
    \]

    The result follows similarly for Case 3 in the proof of \Cref{lemma:entries-of-inverse}.
\end{proof}

It remains to recursively apply this to prove
the overall guarantee of the recursive inversion
algorithm in Figure~\ref{fig:RecInvertWithExcess}.

\subsection{Approximate Inversion Using Excess Vector}

In this section, we prove that the algorithm in \Cref{fig:RecInvertWithExcess} finds an approximate inverse of an invertible RDDL matrix.

\begin{figure}[h]
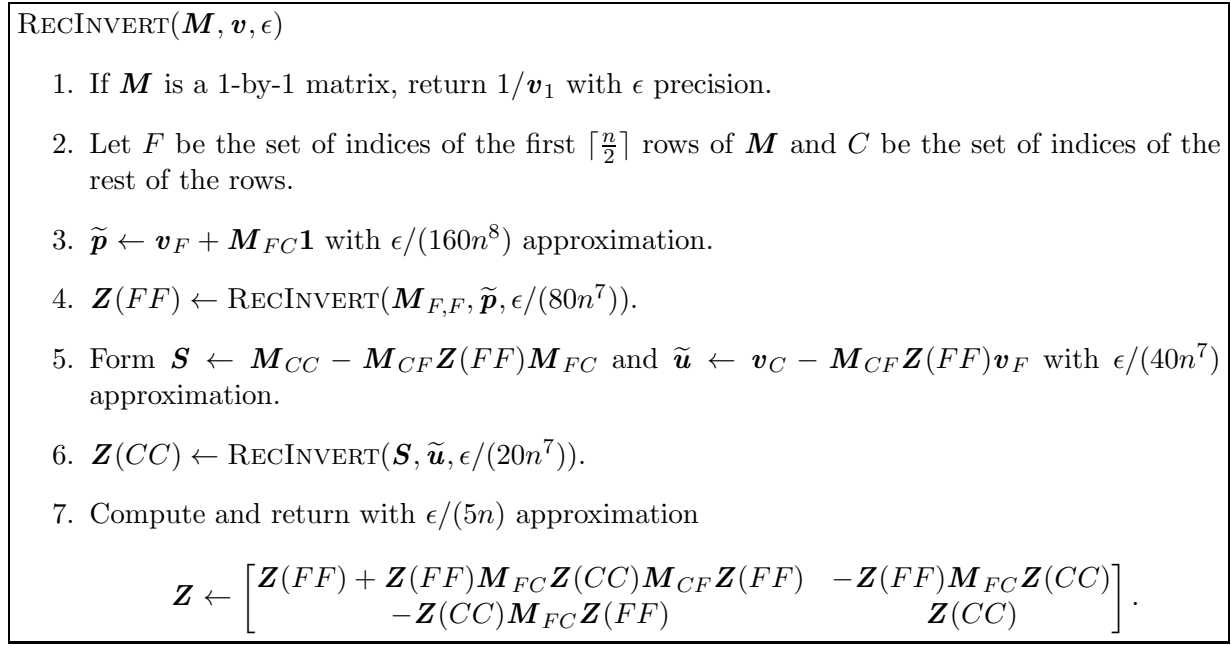

\begin{algbox}
$\textsc{RecInvert}(\MM, \vv, \epsilon)$
\begin{enumerate}
\item If $\MM$ is a $1$-by-$1$ matrix, return $1/\vv_1$ with $\epsilon$ precision. \label{ln:base-induction}
\item Let $F$ be the set of indices of the first $\lceil \frac{n}{2} \rceil$ rows of $\MM$ and $C$ be the set of indices of the rest of the rows.
\item $\widetilde{\pp} \leftarrow \vv_{F} + \MM_{FC} \boldsymbol{1}$ with $\epsilon/(160n^8)$ approximation.
\item $\ZZ(FF) \leftarrow \textsc{RecInvert}(\MM_{F, F}, \widetilde{\pp} ,\epsilon / (80n^7))$. \label{ln:ff-induction}
\item Form $\SS \leftarrow \MM_{CC} - \MM_{CF} \ZZ(FF) \MM_{FC}$ and $\widetilde{\uu} \leftarrow \vv_{C} - \MM_{CF} \ZZ(FF) \vv_{F}$ with $\epsilon/(40 n^7)$ approximation.
\item $\ZZ(CC) \leftarrow  \textsc{RecInvert}(\SS, \widetilde{\uu}, \epsilon / (20n^7))$. \label{ln:cc-induction}
\item Compute and return with $\epsilon/(5n)$ approximation \label{ln:last-line}
\[
\ZZ \leftarrow 
\begin{bmatrix}
\ZZ(FF) + \ZZ(FF) \MM_{FC} \ZZ(CC) \MM_{CF} \ZZ(FF)
&-\ZZ(FF) \MM_{FC} \ZZ(CC)\\
-\ZZ(CC) \MM_{FC} \ZZ(FF) & \ZZ(CC)
\end{bmatrix}.
\]
\end{enumerate}
\end{algbox}
\caption{Psuedocode for recursive inversion
with varying precision}
\label{fig:RecInvertWithExcess}
\end{figure}

\recursionTheorem*
\begin{proof}
We prove the theorem by induction. The base case for a $1$-by-$1$ matrix trivially follows from the construction of Line \ref{ln:base-induction} of the algorithm.
    We now prove $\ZZ(FF) \approxbar_{\epsilon/(40n^7)} \left(\NN_{FF}\right)^{-1}$. Let $\pp = \vv_{F} + \MM_{FC} \boldsymbol{1}$ and $\widetilde{\pp}$ be the floating point approximation of $\pp$. We have $\pp \approxbar_{\epsilon/(160 n^8)} \widetilde{\pp}$.
    
    Let $\widetilde{\NN}_{FF}$ be the matrix such that
    \[
    \widetilde{\NN}_{ij}
    =
    \begin{cases}
    \NN_{ij}, & \text{for $i\neq j$},\\
    \widetilde{\pp}_i + \sum_{\widehat{j} \in F\setminus\{i\}} \MM_{i\widehat{j}}, & \text{for $i = j$}.
    \end{cases}
    \]
    
    First note that $\widetilde{\NN}_{FF} \approxbar_{\epsilon/(160n^8)} \NN_{FF}$ and also row sums approximate each other.
    Therefore by \Cref{lem:ApproxInvert}, $(\widetilde{\NN}_{FF})^{-1} \approxbar_{\epsilon/(80n^7)} \left(\NN_{FF}\right)^{-1}$.
    Moreover by induction
    on Line \ref{ln:ff-induction} of \textsc{RecInvert}, we have
    \[
    \ZZ(FF) \approxbar_{\epsilon/(80n^8)} \left(\widetilde{\NN}_{FF}\right)^{-1}.
    \]
    Therefore, we have 
    \[
    \ZZ(FF) \approxbar_{\epsilon/(40 n^7)} \left(\NN_{FF}\right)^{-1}.
    \]

    Now we show 
    \[
    \ZZ(CC) \approxbar_{\epsilon/(5n^5)} \sc(\NN,C)^{-1}.
    \]
    Let $\TT$ be a matrix given by
    \[
    \TT_{ij}
    =
    \begin{cases}
    \SS_{ij} \qquad \text{for $i\neq j$}\\
    -\left(\widetilde{\uu}_{i} + \sum_{\widehat{j}
    \in C \setminus \left\{i\right\}} \SS_{i\widehat{j}}\right)
    \qquad \text{for $i = j$}.
    \end{cases}
    \]
    Then by induction on Line \ref{ln:cc-induction} of \textsc{RecInvert}, we have
    $\ZZ(CC) \approxbar_{\epsilon/(20 n^7)} \TT^{-1}$. Moreover for $i\neq j$, $\TT_{ij} \approx_{\epsilon/(20 n^7)} \sc(\NN,C)_{ij}$.
        Also for each $i\in C$, we have
    \begin{align*}
    \sc(\NN,C)_{ii}
    & = \NN_{ii} - \NN_{i F} \NN_{FF}^{-1} \NN_{F i} 
    \\
    & =
    -\left(\vv_i + \sum_{j\in[n]\setminus\{i\}} \MM_{ij}\right)
    - \NN_{i F} \NN_{FF}^{-1} \NN_{F i}
    \\ & =
    -\left(\vv_i + \sum_{j\in[n]\setminus\{i\}} \NN_{ij}\right)
    - \NN_{i F} \NN_{FF}^{-1} \NN_{F i}
    \end{align*}
    where the approximate we compute satisfies
    \begin{align*}
        \TT_{ii} & = -\left(\widetilde{\uu}_{i} + \sum_{j\in C \setminus \{i\}} \SS_{ij}\right)
        \\ &
        \approx_{\epsilon/(40n^7)} -\left(\vv_i - \MM_{i F} \ZZ(FF) \vv_{F} + \sum_{j \in C \setminus \{i\}} \left(\MM_{i j} - \MM_{i F} \ZZ(FF) \MM_{F j}\right)\right)
        \\ & \approx_{\epsilon/(40n^7)} 
        - \left(\vv_i - \NN_{i F} \NN_{FF}^{-1} \vv_{F} + \sum_{j \in C \setminus \{i\}} \left(\NN_{i j} - \NN_{i F} \NN_{FF}^{-1} \NN_{F j}\right)\right)
        \\ & = 
        - \left(\vv_i - \NN_{i F} \NN_{FF}^{-1} \left(-\NN_{F [n]}\boldsymbol{1}\right) + \sum_{j \in C \setminus \{i\}} \left(\NN_{i j} - \NN_{i F} \NN_{FF}^{-1} \NN_{F j}\right)\right)
        \\ & =
        - \left(\vv_i + \NN_{i F} \NN_{FF}^{-1} \NN_{F F}\boldsymbol{1} + \NN_{i F} \NN_{FF}^{-1} \NN_{F i} + \sum_{j \in C \setminus \{i\}} \NN_{i j}\right)
        \\ & = 
        - \left(\vv_i + \NN_{i F}\boldsymbol{1} + \NN_{i F} \NN_{FF}^{-1} \NN_{F i} + \sum_{j \in C \setminus \{i\}} \NN_{i j}\right)
        \\ & = 
        -\left(\vv_i + \sum_{j\in[n]\setminus\{i\}} \NN_{ij}\right) - \NN_{i F} \NN_{FF}^{-1} \NN_{F i}
    \end{align*}
Therefore,
\[
\sc(\NN,CC)_{ii} \approx_{\epsilon/(20n^7)} \TT_{ii}.
\]

Furthermore,
\begin{align*}
\TT_{ii} + \sum_{j\in C \setminus \left\{i\right\}} \SS_{ij}
& = - \widetilde{\uu}_i
\\
& \approx_{\epsilon/(40n^7)}
-\left(\vv_i - \MM_{i F} \ZZ(FF) \vv_{F}\right)
\\
& \approx_{\epsilon/(40n^7)}
    -\left(\vv_i - \NN_{i F} \NN_{FF}^{-1} \vv_{F}\right)
\\
& =
    -\left(\vv_i - \NN_{i F} \NN_{FF}^{-1} \left(-\NN_{F \left[n\right]}\boldsymbol{1}\right)\right)
\\
& =
-\left(\vv_i + \NN_{i F}\boldsymbol{1}
+ \NN_{i F} \NN_{FF}^{-1} \NN_{F C}\boldsymbol{1}\right)
\\
& =
-\left(\left(-\NN_{i \left[n\right]} \boldsymbol{1}\right)
+ \NN_{i F}\boldsymbol{1}
+ \NN_{i F} \NN_{FF}^{-1} \NN_{F C}\boldsymbol{1}\right)
\\
& = 
-\left(-\NN_{i C} \boldsymbol{1}
+ \NN_{i F} \NN_{FF}^{-1} \NN_{F C}\boldsymbol{1}\right)
\\
& = 
\NN_{i C} \boldsymbol{1}
- \NN_{i F} \NN_{FF}^{-1} \NN_{F C}\boldsymbol{1}.
\end{align*}

Also,
\[
\sc\left(\NN,C\right)_{ii}
+
\sum_{j\in C \setminus \left\{i\right\}}
\sc\left(\NN,C\right)_{ij}
=
\sum_{j\in C} \NN_{ij}
-
\NN_{i F} \NN_{FF}^{-1} \NN_{F j},
\]
so the row sums also approximate each other with a factor of $e^{\epsilon/(20 n^7)}$.
Therefore by \Cref{lem:ApproxInvert},
\[
\TT^{-1} \approxbar_{\epsilon/(10n^6)} \sc\left(\NN,C\right)^{-1},
\]
and thus $\ZZ(CC) \approxbar_{\epsilon/(5n^6)} \sc(\NN,C)^{-1}$.
Taking these into account for the computation of $\ZZ$ in Line \ref{ln:last-line}, we have $\ZZ \approxbar_{\epsilon / n} \NN^{-1}$.

Note that since our recursion goes for at most $O(\log n)$ iterations, the required accuracy at the lowest level is $\frac{\epsilon}{n^{O(\log n)}}$. Therefore it is enough to work with numbers with $O(\log(\frac{1}{\epsilon}) + \log^2 n)$. The number of arithmetic operations for computing matrix multiplications is also $O(n^3 \log(n))$. Therefore, taking the bit complexity of the input into account, the total number of bit operations is $\Otil(n^3 \cdot (L + \log(\frac{1}{\epsilon})))$.
\end{proof}

We are now ready to prove the main theorem about the computation of escape probabilities.

\mainTheorem*
\begin{proof}
Let $\AA$ be the matrix associated with the Markov chain (random walk) associated with graph $G$.
Without loss of generality suppose the index of $t$ and $p$ in the matrix are $n-1$ and $n$. 
By \eqref{eq:escape-formula}, we need to compute 
\[
- \left(\II - \AA_{1:\left(n - 2\right),1:\left(n-2\right)}\right)^{-1}
\AA_{1:\left(n - 2\right),n}.
\]
Therefore the first part of the algorithm is to call \textsc{RecInvert} procedure.
\[
\XX \leftarrow \textsc{RecInvert}(\II - \AA_{1:(n-2),1:(n-2)}, \AA_{1:(n-2),(n-1)} + \AA_{1:(n-2),n}, \frac{\epsilon}{2}).
\]
Then the algorithm returns $-\XX \AA_{1:(n-2),(n-1)}$ as the solution. Note that all entries of $\XX$ are nonnegative and all entries $-\XX \AA_{1:(n-2),(n-1)}$ are nonpositive. Therefore if we perform the floating-point operations with $\epsilon/(2n+2)$ accuracy, the output vector is within a factor of $\epsilon/2$ of $-\XX \AA_{1:(n-2),(n-1)}$. Combining this with the error bound of $\XX$ gives the result. The running time directly follows from \Cref{thm:main-recursion}.
\end{proof}

\bibliographystyle{alpha}
\bibliography{ref}

\end{document}